\documentclass[11pt]{article}

\newif\ifplanar
\planarfalse

\usepackage[margin=1in]{geometry}
\usepackage{graphicx}
\usepackage[hyphens]{url}
\usepackage{hyperref}
\usepackage{subcaption}
\usepackage{amsmath}
\usepackage{amssymb}
\usepackage{amsthm}
\usepackage{enumerate}
\usepackage[mathletters]{ucs}
\usepackage[utf8x]{inputenc}
\usepackage{xargs}
\usepackage{xspace}
\usepackage[dvipsnames]{xcolor}
\usepackage{tikz}

\let\realbfseries=\bfseries
\def\bfseries{\realbfseries\boldmath}

\usetikzlibrary{shapes,arrows.meta,calc,circuits.ee.IEC}

\tikzset{
  custom cloud/.style={draw,cloud,cloud puffs=9.7},
  coin/.style={fill,circle,minimum size=2.5mm,inner sep=0},
  ground/.style={draw,shape=ground IEC,minimum height=4mm,rotate=-90},
  top ground/.style={ground,rotate=180},
  string label/.style={auto=right},
  string/.style={draw,thick,edge node=node [string label]{#1}},
  optional string/.style={string=#1,dashed},
  rope/.style={string=#1,double,double distance=1mm,shorten <=-0.3mm,shorten >=-0.3mm},
  var/.style={draw,circle},
  wire arrow/.pic={
    \path[tips,every arrow/.try,-{Latex[scale=3,open,fill=white]}] +(0,0) -- +(0.4,0)
    node at (-0.1,0) {\footnotesize#1};
  },
  wire/.style={
    string,
    sloped,allow upside down,
    edge node=pic {wire arrow=#1},
  },
  no wire/.style={string,dotted},
  clause/.style={draw,thick,fill,fill opacity=0.15,text opacity=1,
                 star,star point ratio=2,inner sep=0.5mm},
}%

\graphicspath{{./figures/}}

\newtheorem{theorem}{Theorem}[section]
\newtheorem{lemma}[theorem]{Lemma}
\newtheorem{corollary}{Corollary}[theorem]
\theoremstyle{definition}
\newtheorem{definition}[theorem]{Definition}

{\makeatletter
 \gdef\xxxmark{%
   \expandafter\ifx\csname @mpargs\endcsname\relax 
     \expandafter\ifx\csname @captype\endcsname\relax 
       \marginpar{xxx}
     \else
       xxx 
     \fi
   \else
     xxx 
   \fi}
 \gdef\xxx{\@ifnextchar[\xxx@lab\xxx@nolab}
 \long\gdef\xxx@lab[#1]#2{{\bf [\xxxmark #2 ---{\sc #1}]}}
 \long\gdef\xxx@nolab#1{{\bf [\xxxmark #1]}}
}

\newcommand{\mathsc}[1]{{\normalfont\textsc{#1}}}

\def\defn#1{\textbf{\textit{\boldmath #1}}}

\newcommandx{\nimstring}[1][1=G]{\mathsc{Nimstring}(#1)\xspace}
\newcommandx{\sac}[1][1=G]{\mathsc{Strings-}\allowbreak\mathsc{And-}\allowbreak \mathsc{Coins}(#1)\xspace}
\newcommandx{\lava}[1][1=G]{\mathsc{Coins-}\allowbreak\mathsc{Are-}\allowbreak \mathsc{Lava}(#1)\xspace}
\newcommandx{\gamesat}[1][1=\mathcal{F}]{\mathsc{Game-SAT}(#1)\xspace}

\title{Strings-and-Coins and Nimstring are PSPACE-complete}

\author{%
  Erik D. Demaine%
    \thanks{MIT Computer Science and Artificial Intelligence Laboratory,
      32 Vassar St., Cambridge, MA 02139, USA,
      \protect\url{{edemaine,diomidov}@mit.edu}}
\and
  Yevhenii Diomidov\footnotemark[1]
}

\date{~}

\begin{document}
\maketitle

\begin{abstract}
  We prove that Strings-and-Coins --- the combinatorial two-player game
  generalizing the dual of Dots-and-Boxes --- is strongly PSPACE-complete
  on multigraphs.
  This result improves the best previous result, NP-hardness,
  argued in \emph{Winning Ways}.
  Our result also applies to the Nimstring variant,
  where the winner is determined by normal play;
  indeed, one step in our reduction is the standard reduction
  (also from \emph{Winning Ways}) from Nimstring to Strings-and-Coins.
\end{abstract}

\begin{center}\it
  In memoriam Elwyn Berlekamp (1940--2019), John H. Conway (1937--2020), \\
  and Richard K. Guy (1916--2020)
\end{center}

\section{Introduction}

Elwyn Berlekamp loved Dots and Boxes.
He wrote an entire book, \emph{The Dots and Boxes Game:
Sophisticated Child's Play} \cite{Berlekamp-2000}
devoted to explaining the mathematical underpinnings of the game,
after they were first revealed in Berlekamp, Conway, and Guy's
classic book \emph{Winning Ways} exploring many such combinatorial games
\cite[ch.~16]{Berlekamp-Conway-Guy-2003-vol3}.
At book signings for both books,%
\footnote{The first author had the honor of playing such a game against
  Elwyn at a book signing on April 13, 2004, at Quantum Books in
  Cambridge, Massachusetts.  Elwyn won.}
and after talks he gave about these topics \cite{Berlekamp-Ross-lecture},
Elwyn routinely played simultaneous exhibitions of Dots and Boxes ---
him against dozens of players, in the style of Chess masters.

As many children will tell you, Dots-and-Boxes is a simple pencil-and-paper
game taking place on an $m \times n$ grid of dots.
Two players alternate drawing edges of the grid, with one special rule:
when a player completes the fourth edge of one or two $1 \times 1$ boxes,
that player gains one or two points, respectively,
and \emph{must} immediately draw another edge (a ``free move''
that is often a blessing and a curse).
The game ends when all grid edges have been drawn;
then the player with the most points wins.
(Draws are possible on boards with an even number of squares.)

An equivalent way to think about Dots-and-Boxes is in the dual
of the grid graph.
Think of each $1 \times 1$ square as a dual vertex or \defn{coin}
worth one point, ``tied down'' by four incident \defn{strings} or
dual edges.  Interior strings connect two coins, while boundary strings
connect a coin to the \defn{ground} (not worth any points).
(Equivalently, boundary edges have only one endpoint.)
Now players alternate cutting (removing) strings, and
when a player \defn{frees} one or two coins (removing the last strings
attached to them), that player gains the corresponding number of points
and must move again.
The game ends when all strings have been cut;
then the player with the most points wins.

\defn{Strings-and-Coins}
\cite[pp.~550--551]{Berlekamp-Conway-Guy-2003-vol3},
\cite[ch.~2]{Berlekamp-2000}
is the generalization of this game to arbitrary graphs,
where vertices represent coins and edges represent strings which
can connect up to two coins (the other endpoints being considered ``ground'').
\defn{Nimstring} \cite[pp.~552--554]{Berlekamp-Conway-Guy-2003-vol3},
\cite[ch.~6]{Berlekamp-2000}
is the closely related game where we modify the win condition to
\defn{normal play}:
the first player unable to move loses.
Nimstring is known to be a special case of Strings-and-Coins,
a fact we use in our results; see Lemma~\ref{lem:nimstring-sac} below.

\paragraph{Related work.}
Dots-and-Boxes, Strings-and-Coins, and Nimstring are surprisingly
intricate games with intricate strategy
\cite{Berlekamp-Conway-Guy-2003-vol3,Berlekamp-2000}.
On the mathematical side, even $1 \times n$ Dots-and-Boxes is largely unsolved
\cite{Guy-Nowakowski-2002,DotsBoxes_GONC4}.

To formalize this difficulty,
\emph{Winning Ways} \cite{Berlekamp-Conway-Guy-2003-vol3}
argued in 1984 that deciding the winner of a Strings-and-Coins position
is \emph{NP-hard} by a reduction from vertex-disjoint cycle packing.
Around 2000, Eppstein \cite{Eppstein-cgt-hard} pointed out that this reduction
can be adapted to apply to Dots-and-Boxes as well;
see \cite{Demaine-Hearn-2009-survey-both}.

This work left some natural open problems, first explicitly posed in 2001
\cite{Demaine-Hearn-2009-survey-both}: are Dots-and-Boxes, Strings-and-Coins,
and Nimstring NP-complete or do they extend into a harder complexity class?
Being bounded two-player games, all three naturally lie within PSPACE;
are they PSPACE-complete?


\paragraph{Results.}
In this paper, we settle two out of three of these 20-year-old open problems
by proving that Strings-and-Coins and Nimstring are PSPACE-complete.
This is the first improvement beyond NP-hardness since the original
\emph{Winning Ways} result from 1984.
Our reductions from Game SAT are relatively simple but subtle.
Along the way, we prove PSPACE-completeness of a new Strings-and-Coins
variant called \defn{Coins-Are-Lava},
where the first person to free a coin loses.

Our constructed game positions rely on \defn{multigraphs}
with multiple copies of some edges/strings,
a feature not present in instances corresponding to Dots-and-Boxes.
Thus our results do not apply to Dots-and-Boxes.
A generalization of Dots-and-Boxes that we might be able to target is
\defn{weighted} Dots-and-Boxes, where each grid edge has a specified
number of times it must be drawn before it is ``complete'' and thus
can form the boundary of a $1 \times 1$ box.
This game corresponds to Strings-and-Coins on planar multigraphs
whose vertices can be embedded at grid vertices
such that edges have unit length.
\ifplanar
\else
However, our multigraphs are neither planar nor maximum-degree-$4$,
so they cannot be drawn on a square grid,
so our approach does not resolve the complexity of weighted Dots-and-Boxes.
\fi

In independent work, Buchin, Hagedoorn, Kostitsyna, and van Mulken
\cite{dots-boxes-PSPACE} proved that (unweighted) Dots-and-Boxes is
PSPACE-complete by a reduction from $G_{pos}(\text{POS CNF})$
\cite{Schaefer-1977} (roughly the same problem that we reduce from,
$G_{pos}(\text{POS DNF})$ \cite{Schaefer-1977}).
They construct an instance where, after variable setting, one player's
winning strategy is to select a maximum set of disjoint cycles.
This approach works well for Dots-and-Boxes (and thus Strings-and-Coins)
where the goal is to maximize score,
but not for Nimstring like our approach does.
Thus the two approaches are incomparable.

\section{Nimstring}

We begin with more formal definitions of the games of interest,
and some known lemmas about them:

\begin{definition}[Coin--String Multigraph]
  A \defn{multigraph} $G$ consists of vertices,
  also called \defn{coins}, and edges, also called \defn{strings},
  where each edge $e \in E$ is a set of at most two vertices in~$V$.
  Notably, we allow edges incident to zero or one vertices in~$V$;
  we view the missing endpoints of such an edge as being connected
  to the \defn{ground}.
\end{definition}

\begin{definition}[\unboldmath $\sac$ and $\nimstring$]
  Games $\sac[G]$ and $\nimstring[G]$ are played on a multigraph $G$
  by two players who alternate removing edges and, if a player
  \defn{frees} one or two coins by removing their last incident edges,
  then that player gains the corresponding number of points (one or two)
  and must move again.
  The games end when there are no more strings;
  in $\sac[G]$, the player with the most points wins, while in $\nimstring[G]$,
  the first player unable to move loses.
\end{definition}

Next we prove the standard result that Nimstring is equivalent to a special
case of Strings-and-Coins, and thus hardness of the former implies hardness of
the latter:

\begin{lemma} \label{lem:nimstring-sac}
  {\rm \cite[p.~552]{Berlekamp-Conway-Guy-2003-vol3}}
  For every graph $G$, there exists an efficiently computable graph $H$ such that the winner of $\nimstring[G]$ is the same as the winner of $\sac[H]$.
\end{lemma}
\begin{proof}
  Let $H = G \cup C_{n}$ where $C_{n}$ is a cycle on $n > |V(G)|$ vertices.
  If a player cuts any string in this cycle, then the opponent can claim $n > \frac{|V(H)|}{2}$ coins in a single turn, winning the game. Therefore the players will try to only cut edges in $G$, and the player who cannot do so loses. This goal is equivalent to just playing $\nimstring[G]$.
\end{proof}

A final known result we will need is about ``loony'' positions in Nimstring:

\begin{lemma} \label{lem:loony}
  {\rm \cite[p.~557]{Berlekamp-Conway-Guy-2003-vol3}}
  If $G$ has a degree-$2$ vertex adjacent to exactly one degree-$1$ vertex, then the first player can always win in $\nimstring[G]$. Such positions are known as \defn{loony} positions.
\end{lemma}

\begin{figure}[ht]
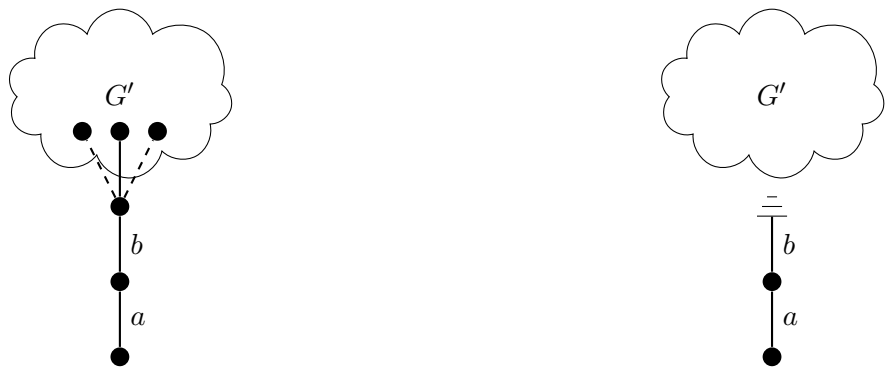

  \begin{subfigure}[t]{0.5\linewidth}
    \centering
    \tikz \draw
    (0,3.5) node[custom cloud,minimum width=30mm,minimum height=22.5mm] {$G'$}
    node[coin] (A) at (0,0) {}
    node[coin] (B) at (0,1) {}
    node[coin] (C) at (0,2) {}
    node[coin] (D1) at (-0.5,3) {}
    node[coin] (D2) at (0.0,3) {}
    node[coin] (D3) at (0.5,3) {}
    (A) edge[string=$a$] (B)
    (B) edge[string=$b$] (C)
    (C) edge[optional string] (D1) edge[string] (D2) edge[optional string] (D3);
    \caption{String $b$ connects to a vertex of degree at least $2$.}
    \label{fig:loony-1}
  \end{subfigure}
  \hfill
  \begin{subfigure}[t]{0.45\linewidth}
    \centering
    \tikz \draw
    (0,3.5) node[custom cloud,minimum width=30mm,minimum height=22.5mm] {$G'$}
    node[coin] (A) at (0,0) {}
    node[coin] (B) at (0,1) {}
    node[top ground] (C) at (0,2) {}
    (A) edge[string=$a$] (B)
    (B) edge[string=$b$] (C);
    \caption{String $b$ connects to the ground.}
    \label{fig:loony-2}
  \end{subfigure}
  \caption{Two loony positions}
  \label{fig:loony}
\end{figure}

\begin{proof}
  Let $a$ be the string between the two coins, $b$ be the other string connected to a degree-2 coin, and $G'$ be the rest of the graph (Figure~\ref{fig:loony}). One of the players has a winning strategy in $\nimstring[G']$.
  \begin{itemize}
    \item If the first player has a winning strategy in $\nimstring[G']$, then we cut strings $a$ and $b$ in this order. We get exactly $G'$ and it is still our turn. By assumption we can win.
    \item If the second player has a winning strategy in $\nimstring[G']$, then we just cut string $b$. We get graph $G'$ (plus an extra edge that does not affect the game), and it is our opponent's turn. By assumption opponent cannot win.
    \qedhere
  \end{itemize}
\end{proof}

\section{Coins-are-Lava}

We introduce a variant game played on strings and coins
that we find easier to analyze, called \defn{Coins-are-Lava}:%
\footnote{%
  For a ``practical'' motivation for this game,
  consider the 1933 Double Eagle U.S. coin: until 2002,
  possession of this coin could result in imprisonment \cite{US-mint}.
}

\begin{definition}[\lava]
  Game $\lava[G]$ is played on a multigraph $G$ by
  two players who alternate removing edges and, if a player frees a coin,
  that player loses.
  Equivalently, players are forbidden from removing an edge that would free
  a coin, and the winner is determined according to normal play.
\end{definition}

Now we show that Coins-are-Lava is a special case of Nimstring.
Thus, its hardness will imply the hardness of both Nimstring and
(by Lemma~\ref{lem:nimstring-sac}) Strings-and-Coins.

\begin{lemma}
  \label{lem:lava-nimstring}
  For every graph $G$, there exists an efficiently computable graph $H$ such that the winner of $\lava[G]$ is the same as the winner of $\nimstring[H]$.
\end{lemma}

\begin{figure}[ht]
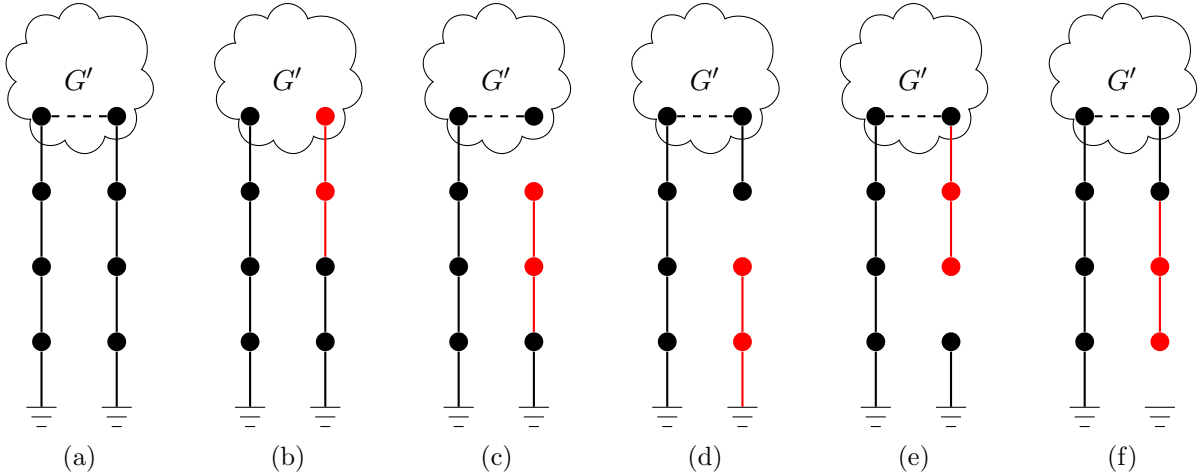

  \begin{subfigure}[t]{0.16\linewidth}
    \centering
    \tikz \draw
    (0.5,0.5) node[custom cloud,minimum width=20mm,minimum height=20mm] {$G'$}
    node[coin] (A1) at (0,-0) {} node[coin] (A2) at (1,-0) {}
    node[coin] (B1) at (0,-1) {} node[coin] (B2) at (1,-1) {}
    node[coin] (C1) at (0,-2) {} node[coin] (C2) at (1,-2) {}
    node[coin] (D1) at (0,-3) {} node[coin] (D2) at (1,-3) {}
    node[ground] (E1) at (0,-4) {} node[ground] (E2) at (1,-4) {}
    (A1) edge[optional string] (A2)
    (A1) edge[string] (B1) (A2) edge[string] (B2)
    (B1) edge[string] (C1) (B2) edge[string] (C2)
    (C1) edge[string] (D1) (C2) edge[string] (D2)
    (D1) edge[string] (E1) (D2) edge[string] (E2);
    \caption{}
    \label{fig:lava-nimstring-1}
  \end{subfigure}
  \hfill
  \begin{subfigure}[t]{0.16\linewidth}
    \centering
    \tikz \draw
    (0.5,0.5) node[custom cloud,minimum width=20mm,minimum height=20mm] {$G'$}
    node[coin] (A1) at (0,-0) {} node[coin,red] (A2) at (1,-0) {}
    node[coin] (B1) at (0,-1) {} node[coin,red] (B2) at (1,-1) {}
    node[coin] (C1) at (0,-2) {} node[coin] (C2) at (1,-2) {}
    node[coin] (D1) at (0,-3) {} node[coin] (D2) at (1,-3) {}
    node[ground] (E1) at (0,-4) {} node[ground] (E2) at (1,-4) {}
    (A1) edge[string] (B1) (A2) edge[string,red] (B2)
    (B1) edge[string] (C1) (B2) edge[string,red] (C2)
    (C1) edge[string] (D1) (C2) edge[string] (D2)
    (D1) edge[string] (E1) (D2) edge[string] (E2);
    \caption{}
  \end{subfigure}
  \hfill
  \begin{subfigure}[t]{0.16\linewidth}
    \centering
    \tikz \draw
    (0.5,0.5) node[custom cloud,minimum width=20mm,minimum height=20mm] {$G'$}
    node[coin] (A1) at (0,-0) {} node[coin] (A2) at (1,-0) {}
    node[coin] (B1) at (0,-1) {} node[coin,red] (B2) at (1,-1) {}
    node[coin] (C1) at (0,-2) {} node[coin,red] (C2) at (1,-2) {}
    node[coin] (D1) at (0,-3) {} node[coin] (D2) at (1,-3) {}
    node[ground] (E1) at (0,-4) {} node[ground] (E2) at (1,-4) {}
    (A1) edge[optional string] (A2)
    (A1) edge[string] (B1)
    (B1) edge[string] (C1) (B2) edge[string,red] (C2)
    (C1) edge[string] (D1) (C2) edge[string,red] (D2)
    (D1) edge[string] (E1) (D2) edge[string] (E2);
    \caption{}
  \end{subfigure}
  \hfill
  \begin{subfigure}[t]{0.16\linewidth}
    \centering
    \tikz \draw
    (0.5,0.5) node[custom cloud,minimum width=20mm,minimum height=20mm] {$G'$}
    node[coin] (A1) at (0,-0) {} node[coin] (A2) at (1,-0) {}
    node[coin] (B1) at (0,-1) {} node[coin] (B2) at (1,-1) {}
    node[coin] (C1) at (0,-2) {} node[coin,red] (C2) at (1,-2) {}
    node[coin] (D1) at (0,-3) {} node[coin,red] (D2) at (1,-3) {}
    node[ground] (E1) at (0,-4) {} node[ground] (E2) at (1,-4) {}
    (A1) edge[optional string] (A2)
    (A1) edge[string] (B1) (A2) edge[string] (B2)
    (B1) edge[string] (C1)
    (C1) edge[string] (D1) (C2) edge[string,red] (D2)
    (D1) edge[string] (E1) (D2) edge[string,red] (E2);
    \caption{}
  \end{subfigure}
  \hfill
  \begin{subfigure}[t]{0.16\linewidth}
    \centering
    \tikz \draw
    (0.5,0.5) node[custom cloud,minimum width=20mm,minimum height=20mm] {$G'$}
    node[coin] (A1) at (0,-0) {} node[coin] (A2) at (1,-0) {}
    node[coin] (B1) at (0,-1) {} node[coin,red] (B2) at (1,-1) {}
    node[coin] (C1) at (0,-2) {} node[coin,red] (C2) at (1,-2) {}
    node[coin] (D1) at (0,-3) {} node[coin] (D2) at (1,-3) {}
    node[ground] (E1) at (0,-4) {} node[ground] (E2) at (1,-4) {}
    (A1) edge[optional string] (A2)
    (A1) edge[string] (B1) (A2) edge[string,red] (B2)
    (B1) edge[string] (C1) (B2) edge[string,red] (C2)
    (C1) edge[string] (D1)
    (D1) edge[string] (E1) (D2) edge[string] (E2);
    \caption{}
  \end{subfigure}
  \hfill
  \begin{subfigure}[t]{0.16\linewidth}
    \centering
    \tikz \draw
    (0.5,0.5) node[custom cloud,minimum width=20mm,minimum height=20mm] {$G'$}
    node[coin] (A1) at (0,-0) {} node[coin] (A2) at (1,-0) {}
    node[coin] (B1) at (0,-1) {} node[coin] (B2) at (1,-1) {}
    node[coin] (C1) at (0,-2) {} node[coin,red] (C2) at (1,-2) {}
    node[coin] (D1) at (0,-3) {} node[coin,red] (D2) at (1,-3) {}
    node[ground] (E1) at (0,-4) {} node[ground] (E2) at (1,-4) {}
    (A1) edge[optional string] (A2)
    (A1) edge[string] (B1) (A2) edge[string] (B2)
    (B1) edge[string] (C1) (B2) edge[string,red] (C2)
    (C1) edge[string] (D1) (C2) edge[string,red] (D2)
    (D1) edge[string] (E1);
    \caption{}
  \end{subfigure}
  \caption{(a) The graph $H$; (b) freeing a coin in $G$ results in a loony position; (c--f) cutting a string outside $G$ results in a loony position.}
  \label{fig:lava-nimstring}
\end{figure}

\begin{proof}
  Let $H$ be a graph obtained from $G$ by connecting every coin to the ground with a long chain (length $\geq 5$); see Figure~\ref{fig:lava-nimstring-1}.
  
  If a player cuts a string in one of these chains, or cuts all strings in $G$ attached to the same coin, this creates a loony position and ends their turn; see Figure~\ref{fig:lava-nimstring}.
  By Lemma~\ref{lem:loony}, their opponent can then win.

  Therefore the players will try to avoid cutting strings outside $G$ or freeing a coin in $G$.
  The first player to fail to do so loses.
  This goal is equivalent to $\lava[G]$.
\end{proof}

\section{PSPACE-Hardness}

It remains to prove that $\lava$ is PSPACE-complete.
Our reduction is from the following known PSPACE-complete problem.

\begin{definition}[$\gamesat$]
  Given a positive DNF formula $\mathcal F$ (an \textsc{or} of \textsc{and}s
  of variables without negation), $\gamesat$ is the following game played by
  two players, Trudy and Fallon.
  Initially each variable is \defn{unset}.
  In each turn, the player may \defn{set} a variable to \textsc{true} or
  \textsc{false}, or the player may \defn{skip} their turn (do nothing).
  The game ends when all variables are set; then
  Trudy wins if formula $\mathcal F$ is true, while
  Fallon wins if formula $\mathcal F$ is false.
\end{definition}

We allow players to skip turns and to set variables to the ``wrong'' value
(Trudy to false or Fallon to true).  The player with a winning strategy can
always avoid such moves, however, replacing them with dominating ``good'' moves
that do not skip and play the ``right'' value (Trudy to true or Fallon to
false), as such moves never hurt the winning player's final goal.

Schaefer \cite{Schaefer-1977} proved that this game is PSPACE-complete,
under the name $G_{pos}(\text{POS DNF})$.

\begin{theorem}
  \label{thm:lava-pspace}
  Coins-are-Lava is PSPACE-complete.
\end{theorem}

\begin{proof}
  Let $\mathcal{F}$ be a positive DNF formula with $n$ variables, $m$ clauses, and $k_i$ occurrences of each variable $x_i$.
  Without loss of generality, every clause contains at least $2$ variables and every variable appears in at least $1$ clause.
  Fix a sufficiently large number $N \gg m^2n^2$.

  First we define several useful gadgets, which will be connected together via shared coins (merging the output coin of one gadget with the input coin of another gadget).
  Many of these gadgets are parametrized by an integer \emph{level}.
  Intuitively, doing anything to a level-$(\ell+1)$ gadget requires an order of magnitude more time than doing anything to a level-$\ell$ gadget.
  This way we can make sure that players interact with gadgets in the right order.
  However since each level-$\ell$ gadget uses $N^{θ(\ell)}$ strings, we can only use a constant number of levels.

  \begin{figure}[ht]
    \centering
    \begin{minipage}[b]{0.3\linewidth}
      \centering
      \tikz \draw
      node[coin] (A1) at (-1,-1.5) {} node[coin] (B1) at (-1,1.5) {}
      (A1) edge[rope=$5$] (B1)
      node at (0,0) {\huge$:=$}
      node[coin] (A2) at (1.5,-1.5) {} node[coin] (B2) at (1.5,1.5) {}
      (A2) edge[string,bend left=30] (B2) edge[string,bend left=15] (B2) edge[string] (B2) edge[string,bend right=15] (B2) edge[string,bend right=30] (B2);
      \caption{A width-$5$ rope}
      \label{fig:rope}
    \end{minipage}
    \hfill
    \begin{minipage}[b]{0.6\linewidth}
      \centering
      \begin{subfigure}[t]{0.4\linewidth}
        \centering
        \tikz \draw
        node[var] at (-1,0) {$x_i$}
        node at (0,0) {\huge$:=$}
        node[coin,label=above:out] (A) at (1,1) {}
        node[coin] (B) at (1,0) {}
        node[ground] (C) at (1,-1) {}
        (A) edge[string] (B) (B) edge[string] (C);
        \caption{Initial state (unset)}
      \end{subfigure}
      \hfill
      \begin{subfigure}[t]{0.4\linewidth}
        \centering
        \tikz \draw
        node[coin,label=above:out] (A1) at (-1,1) {}
        node[coin] (B1) at (-1,0) {}
        node[ground] (C1) at (-1,-1) {}
        (B1) edge[string] (C1)
        node[coin,label=above:out] (A2) at (1,1) {}
        node[coin] (B2) at (1,0) {}
        node[ground] (C2) at (1,-1) {}
        (A2) edge[string] (B2);
        \caption{Variable set to true and false respectively}
      \end{subfigure}
      \caption{Variable gadget}
      \label{fig:variable}
    \end{minipage}
  \end{figure}

  A \defn{rope} (Figure~\ref{fig:rope}) is a collection of strings that share both endpoints.
  The number of strings in a rope is called its \defn{width}.
  We say that a rope has been \defn{cut} when all of its strings have been cut.
  When the game ends, every rope has either been cut completely, or it has only 1 string remaining.
  (Otherwise, a string in the rope can always be safely cut without freeing any coin.)

  A \defn{variable} gadget (Figure~\ref{fig:variable}) consists of a chain of two strings, where the bottom string is connected to the ground and the top string is connected to an output coin.
  We say that is \defn{set to false} if the bottom string has been cut, \defn{set to true} if the top string has been cut, and \defn{unset} if neither string has been cut.
  A variable implicitly has level $0$.

  \begin{figure}
    \centering
    \begin{minipage}[b]{0.6\linewidth}
      \centering
      \begin{subfigure}[t]{0.4\linewidth}
        \centering
        \tikz \draw
        node[coin] (A1) at (-1,1) {}
        node[coin] (C1) at (-1,-1) {}
        (C1) edge[wire=$\ell\,$] (A1)
        node at (0,0) {\huge$:=$}
        node[coin,label=above:out] (A2) at (1,1) {}
        node[coin] (B2) at (1,0) {}
        node[coin,label=below:in] (C2) at (1,-1) {}
        (C2) edge[rope=$N^{2\ell-1}$] (B2)
        (B2) edge[rope=$N^{2\ell}$] (A2);
        \caption{Initial state}
      \end{subfigure}
      \hfill
      \begin{subfigure}[t]{0.5\linewidth}
        \centering
        \tikz \draw
        node[coin,label=above:out] (A1) at (-1,1) {}
        node[coin] (B1) at (-1,0) {}
        node[coin,label=below:in] (C1) at (-1,-1) {}
        (B1) edge[rope] (A1)
        node[coin,label=above:out] (A2) at (1,1) {}
        node[coin] (B2) at (1,0) {}
        node[coin,label=below:in] (C2) at (1,-1) {}
        (C2) edge[rope] (B2);
        \caption{Disabled and activated wires}
      \end{subfigure}
      \caption{Wire gadget}
      \label{fig:wire}
    \end{minipage}
    \hfill
    \begin{minipage}[b]{0.3\linewidth}
      \centering
      \tikz \draw
      node[clause] at (-1,0) {$\ell$}
      node at (0,0) {\huge$:=$}
      node[top ground,scale=1.5] (A) at (1,1) {}
      node[coin,label=below:in] (B) at (1,-1) {}
      (B) edge[rope=$N^{2\ell-1}$] (A);
      \caption{Clause gadget}
      \label{fig:clause}
    \end{minipage}
  \end{figure}

  A level-$\ell$ \defn{wire} gadget (Figure~\ref{fig:wire}) consists of a chain of two ropes, a width-$N^{2\ell-1}$ bottom rope connected to an input coin and a width-$N^{2\ell}$ top rope connected to an output coin.
  We say that it is \defn{disabled} if the input rope has been cut, \defn{activated} if the top rope has been cut.
  The \defn{HP} (Hit Points) of the wire is the number of strings remaining in the bottom rope.
  Note that activating a wire takes a factor of $N$ more moves than disabling it.
  This means that, if one player is racing to activate a wire and the other is racing to disable it, then the disabler will win the race.
  Intuitively, the only case where a wire will get activated is if disabling the wire would free a coin.

  A level-$\ell$ \defn{clause} gadget (Figure~\ref{fig:clause}) consists of a single width-$N^{2\ell-1}$ rope connected to an input coin and the ground.
  We say that it is \defn{disabled} if the rope has been cut.
  The \defn{HP} of the clause is the number of strings remaining in the rope.
  
  The winner is determined solely by the parity of the number of removed strings.
  We can easily flip this parity, for example by adding an extra ground-to-ground string.
  So without loss of generality, Fallon wins if (but not only if) every variable and wire has one string remaining and all $m$ clauses have no strings.
  Then Trudy wins if every variable and wire has one string remaining, $m-1$ clauses have no strings, and the final clause has one string.
  In fact, we will show that the game has to end in one of these two specific ways.
  
  Let $\mathcal{F}'$ be a new formula with the following clauses:
  \begin{quote}
  \begin{itemize}
    \item all clauses from $\mathcal{F}$, which we call \defn{real} clauses;
    \item for every variable, a \defn{singleton} clause containing just that variable; and
    \item one additional \defn{empty} clause that contains no variables and is always satisfied.
  \end{itemize}
  \end{quote}

  \begin{figure}
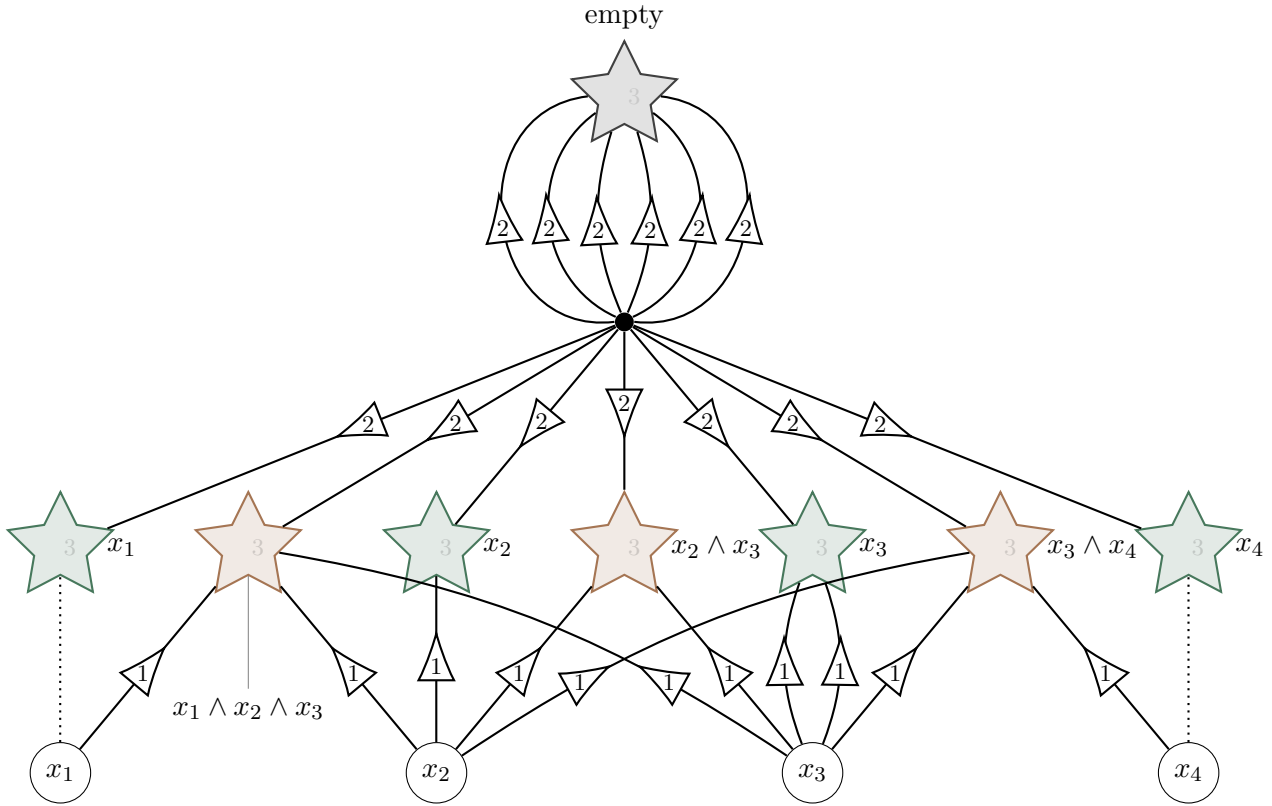

    \centering
    \tikz \path
    node[matrix,row sep={3cm,between origins},column sep={2.5cm,between origins}] {
      &&& \node[clause,gray!50!black,label=above:empty] (E) {\footnotesize$3$}; &&& \\
      &&& \node[coin] (R) {}; &&& \\
      \node[clause,Green!50!black,text=black,label=right:$x_1$] (S1) {\footnotesize$3$}; &
      \node[clause,Orange!60!black,text=black,pin={[pin distance=15mm]below:$x_1 \land x_2 \land x_3$}] (C1) {\footnotesize$3$}; &
      \node[clause,Green!50!black,text=black,label=right:$x_2$] (S2) {\footnotesize$3$}; &
      \node[clause,Orange!60!black,text=black,label=right:$x_2 \land x_3$] (C2) {\footnotesize$3$}; &
      \node[clause,Green!50!black,text=black,label=right:$x_3$] (S3) {\footnotesize$3$}; &
      \node[clause,Orange!60!black,text=black,label=right:$x_3 \land x_4$] (C3) {\footnotesize$3$}; &
      \node[clause,Green!50!black,text=black,label=right:$x_4$] (S4) {\footnotesize$3$}; \\
      \node[var] (X1) {$x_1$}; &&
      \node[var] (X2) {$x_2$}; &&
      \node[var] (X3) {$x_3$}; &&
      \node[var] (X4) {$x_4$}; \\
    }
    (R) edge[wire=$2$] (S1) edge[wire=$2$] (S2) edge[wire=$2$] (S3) edge[wire=$2$] (S4)
        edge[wire=$2$] (C1) edge[wire=$2$] (C2) edge[wire=$2$] (C3)
        edge[wire=$2$,bend left=90,looseness=1.5] (E)
        edge[wire=$2$,bend left=60,looseness=1.1] (E)
        edge[wire=$2$,bend left=20] (E)
        edge[wire=$2$,bend right=20] (E)
        edge[wire=$2$,bend right=60,looseness=1.1] (E)
        edge[wire=$2$,bend right=90,looseness=1.5] (E)
    (X1) edge[no wire] (S1) edge[wire=$1$] (C1)
    (X2) edge[wire=$1$] (S2)
         edge[wire=$1$] (C1) edge[wire=$1$] (C2)
         edge[wire=$1$,pos=0.25,bend left=12] (C3)
    (X3) edge[wire=$1$,bend left=20] (S3)
         edge[wire=$1$,bend right=20] (S3)
         edge[wire=$1$,pos=0.25,bend right=12] (C1)
         edge[wire=$1$] (C2) edge[wire=$1$] (C3)
    (X4) edge[no wire] (S4) edge[wire=$1$] (C3)
    ;
    \caption{Graph $G$ for formula $(x_1 \land x_2 \land x_3) \lor (x_2 \land x_3) \lor (x_3 \land x_4)$. Clauses are labeled and colored according to whether they are empty (``empty'' and gray, at the top), singleton (``$x_i$'' and green), or real (``$x_i \land x_j \cdots$'' and orange). Dotted lines indicate that there are supposed to be $k_i-1$ wires there, but $k_i-1=0$.}
    \label{fig:main-reduction}
  \end{figure}

  We construct a multigraph $G$ by connecting the gadgets as follows;
  refer to Figure~\ref{fig:main-reduction}:
  \begin{quote}
  \begin{itemize}
    \item a variable gadget for each variable;
    \item a clause gadget for each clause;
    \item a level-$1$ wire from each variable $x_i$ to each of $k_i$ real clauses that contain that variable;
    \item $k_i - 1$ level-$1$ wires from each variable to the corresponding singleton clause;
    \item a single vertex called the \defn{root coin};
    \item a level-$2$ wire from the root coin to every real clause and every singleton clause; and
    \item $n+m-1$ level-$2$ wires from the root coin to the empty clause.
  \end{itemize}
  \end{quote}

  First we describe how typical gameplay in $G$ should look (without proofs) to give some intuition for why this construction makes sense, and then we prove that it works more formally.
  Typical gameplay divides into four sequential phases:
  \begin{enumerate}
    \item First Trudy and Fallon set variable gadgets to true and false respectively.
    \item Then the players disable all wires from false variables, and disable all but one wire from each true variable (disabling all wires from a true variable would free a coin).
    Then they activate the level-$1$ wires that have not been disabled (one from each true variable).
    Note that almost half the wires from each variable go to singleton clauses.
    If all real clauses are false, then those wires form the majority, and Fallon can ensure that one of them gets activated.
    But if even one real clause is true, the wires to true clauses (real or singleton) now form a majority, and Trudy can ensure that one of them gets activated.
    \item Then the players disable all but one level-$2$ wire and activate the remaining level-$2$ wire (disabling all of them would free a coin).
    Almost half of these wires go to the empty clause.
    If all real clauses are false, then they form a majority, and Fallon can ensure that one of them gets activated.
    But if even one real clause is true, then it together with the empty clause forms a majority, and Trudy can ensure that one of them gets activated.
    \item Finally, the players disable the clause gadgets.
    A clause can be disabled unless all wires pointing at it got activated (in that case, disabling it would free a coin).
    If the formula is not satisfied, then all clauses get disabled and Fallon wins.
    If the formula is satisfied, then exactly one clause remains and Trudy wins.
  \end{enumerate}
  
  We want to show that the winner of $\lava$ is the same as the winner of $\gamesat$.
  We do a case split on the winner of $\gamesat$, and in each case provide a winning $\lava$ strategy for that player.

  If Fallon can win $\gamesat$, then they can win $\lava$ using the following strategy (where numbers match the phases of intended gameplay above):
  \begin{enumerate}
    \item There is a natural mapping $f$ from states of $\lava$ to states of $\gamesat$: a variable $x_i$ in $\gamesat$ is set to true if the corresponding variable gadget is set to true, set to false if the gadget is set to false, and unset if the gadget is unset.
    Every move in $\lava$ maps to a valid move in $\gamesat$,
    where moves outside of variable gadgets map to skip moves.
    Also, if we played $\lava$ for less than $2n$ moves, then we can perform any move that is valid in the corresponding $\gamesat$ state.
    This does not free a coin, because the relevant coin has degree $Ω(N) \gg 2n$.
    So we can transfer the strategy from $\gamesat$ to $\lava$: for every opponent's move in $\lava$, map it to $\gamesat$, find the best response, and map it back to $\lava$.
    We remain in this phase until we have set all variable gadgets to some assignment that does not satisfy $\mathcal{F}$, as guaranteed by the winning strategy in $\gamesat$.
    \item Call a level-$1$ wire from a true variable $x_i$ \defn{good} if it points at a real clause and \defn{bad} if it points at a singleton clause.
    Wires from false variables are \defn{neutral}.
    Each true variable $x_i$ has $k_i$ good wires and $k_i-1$ bad ones.
    For each true variable $x_i$, the total HP of bad wires is still at most $(k_i-1) N^1$ and total HP of all good wires is at least $k_i N^1 - O(n) > (k_i-1) N^1$ (the opponent could cut up to $O(n)$ strings here while we were setting variables).
    \begin{enumerate}
      \item Disable all bad wires, 
      Specifically, if the opponent reduced HP of a good wire connected to some true variable $x_i$, we respond by reducing HP of a bad wire connected to the same $x_i$; if the opponent did something else or $x_i$ has no bad wires left, we reduce HP of a bad wire connected an arbitrary variable $x_j$.
      This maintains the invariant that for each true variable $x_i$, HP of $x_i$'s good wires is higher than HP of $x_i$'s bad wires.
      The opponent cannot activate any bad wires because that would take $Θ(N^2) \gg \sum_i k_i N^1$ moves.
      \item Disable good and neutral level-$1$ wires until there is only one good wire remaining per true variable.
      Once again, the opponent cannot activate these wires because that would take too many moves.
      \item Activate the remaining good wires.
      Opponent cannot disable these wires, because that would free a coin.
      \item We have activated exactly one good wire per true variable.
      There are no activated level-$1$ wires pointing at satisfied clauses, because real clauses are unsatisfied and singleton clauses are bad.
    \end{enumerate}
    \item Call a level-$2$ wire \defn{good} if it points to a real or singleton clause and \defn{bad} if it points to the empty clause.
    The total HP of the $n+m-1$ bad wires is at most $(n + m - 1) N^3$ and the total HP of the $n+m$ good wires is still (after $O(nmN^2)$ moves spent in the first two stages) at least $(n + m) N^3 - O(nm N^2) > (n + m - 1) N^3$.
    \begin{enumerate}
      \item Disable all bad wires. The opponent cannot disable all good wires before we disable the bad ones because good wires have more HP.
      \item Disable all but one good wire.
      Because all disabling and activating steps done so far are for wires of HP $\Theta(N^3)$,
      and activating a level-$2$ wire requires $Θ(N^4)$ moves,
      the opponent cannot afford to activate any of these good wires
      before we disable them.
      \item Activate the last good wire.
      The opponent cannot disable it because that would free the root coin.
    \end{enumerate}
    \item Disable all clause gadgets.
    This will not free a coin, because every clause has at least one disabled wire: real clauses are unsatisfied so there is a false variable whose adjacent wire we disabled in Step 2(b); singleton clauses have bad level-$1$ wires that we disabled in Step 2(a); and the empty clause has a bad level-$2$ wire that we disabled in Step 3(a).
    We win because there are no clause gadgets remaining.
  \end{enumerate}
  
  If Trudy can win $\gamesat$, then they can win $\lava$ using the following strategy (where numbers match the phases of intended gameplay above):
  \begin{enumerate}
    \item Set the variable gadgets to some assignment that satisfies $\mathcal{F}$.
    Let $C \in \mathcal{F}$ be a satisfied clause.
    \item Call a level-$1$ wire from a true variable $x_i \in C$ \defn{good} if it points at a singleton clause or $C$ and \defn{bad} otherwise.
    Wires from variables not in $C$ are \defn{neutral}.
    Each variable $x_i \in C$ has $k_i - 1$ bad wires ($k_i$ to real clauses, but one of them is $C$) and $k_i$ good ones ($k_i - 1$ to the singleton clause and one to $C$).
    Disable all bad wires, then disable all neutral wires and all-but-one good wire per variable in $C$, and then activate the remaining good wires.
    Each activated wire points to a singleton clause or to $C$.
    Then either all of them point to $C$, or at least one of them points to a singleton clause.
    Either way, we have some satisfied real or singleton clause $C'$ with only activated level-$1$ wires.
    \item Call a level-$2$ wire \defn{good} if it points to $C'$ or to the empty clause.
    There are $n + m - 1$ bad wires ($n + m$ to real clauses, but one of them is $C'$) and $n + m$ bad ones ($n + m - 1$ to the empty clause plus one to $C'$).
    Disable all bad wires and activate exactly one good wire.
    Let $C''$ be the clause pointed by the activated wire (either $C'$ or the empty clause).
    \item Disable all clause gadgets other than $C''$.
    This will not free a coin, because every clause other than $C''$ has a disabled level-$2$ wire.
    But $C''$ cannot be disabled, because all of the wires pointing at it have been activated.
    We win because there is exactly one clause gadget remaining.
    \qedhere
  \end{enumerate}
\end{proof}

\begin{corollary}
  \label{cor:nimstring-pspace}
  Nimstring is PSPACE-complete.
\end{corollary}
\begin{proof}
  This follows immediately from Theorem~\ref{thm:lava-pspace} and Lemma~\ref{lem:lava-nimstring}.
\end{proof}

\begin{corollary}
  Strings-and-Coins is PSPACE-complete.
\end{corollary}
\begin{proof}
  This follows immediately from Corollary~\ref{cor:nimstring-pspace} and Lemma~\ref{lem:nimstring-sac}.
\end{proof}

\section{Open Problems}

We have proved PSPACE-completeness of Strings-and-Coins and Nimstring on
multigraphs, while Buchin et al.\ \cite{dots-boxes-PSPACE} proved
PSPACE-completeness of Dots-and-Boxes, and thus Strings-and-Coins on
grid graphs.
The main open problem is whether Dots-and-Boxes with normal play
instead of scoring, i.e., Nimstring on grid graphs, is also PSPACE-complete.
Toward this goal, we could also aim to prove PSPACE-completeness of Nimstring
on simple graphs (with only one copy of each edge/string) or planar graphs.

\section*{Acknowledgments}

This work was initiated during an MIT class on Algorithmic Lower Bounds:
Fun with Hardness Proofs (6.892, Spring 2019).  We thank the other participants
of the class for providing an inspiring research environment.

\let\realbibitem=\bibitem
\def\bibitem{\par \vspace{-1.2ex}\realbibitem}

\newcommand\Bibkey[3]{} 
\bibliography{paper}

\begin{thebibliography}{BHKvM21}

\bibitem[AMS03]{Berlekamp-Ross-lecture}
American Mathematical~Society\Bibkey AMS.
\newblock {E}lwyn {B}erlekamp gives {A}rnold {R}oss lecture.
\newblock \url{http://www.ams.org/programs/students/arl2004}, 2003.

\bibitem[BCG03]{Berlekamp-Conway-Guy-2003-vol3}
Elwyn~R. Berlekamp, John~H. Conway, and Richard~K. Guy.
\newblock {\em Winning Ways for Your Mathematical Plays}, volume~3.
\newblock A K Peters, Wellesley, MA, 2nd edition, 2003.
\newblock First edition published in 1982.

\bibitem[Ber00]{Berlekamp-2000}
Elwyn Berlekamp.
\newblock {\em The Dots and Boxes Game: Sophisticated Child's Play}.
\newblock A K Peters, Wellesley, MA, 2000.

\bibitem[BHKvM21]{dots-boxes-PSPACE}
Kevin Buchin, Mart Hagedoorn, Irina Kostitsyna, and Max van Mulken.
\newblock Dots {\&} boxes is {PSPACE}-complete.
\newblock arXiv:2105.02837, 2021.
\newblock \url{https://arXiv.org/abs/2105.02837}.

\bibitem[CDDL15]{DotsBoxes_GONC4}
S\'ebastien Collette, Erik~D. Demaine, Martin~L. Demaine, and Stefan Langerman.
\newblock Narrow mis\`ere dots-and-boxes.
\newblock In Richard~J. Nowakowski, editor, {\em Games of No Chance 4}, pages
  57--64. Cambridge University Press, 2015.

\bibitem[DH09]{Demaine-Hearn-2009-survey-both}
Erik~D. Demaine and Robert~A. Hearn.
\newblock Playing games with algorithms: Algorithmic combinatorial game theory.
\newblock In {\em Games of No Chance 3}, pages 3--56. Cambridge University
  Press, 2009.
\newblock arXiv:cc.CC/0106019. First version appeared at MFCS 2001.

\bibitem[Epp]{Eppstein-cgt-hard}
David Eppstein.
\newblock Computational complexity of games and puzzles.
\newblock \url{http://www.ics.uci.edu/~eppstein/cgt/hard.html}.

\bibitem[GN02]{Guy-Nowakowski-2002}
Richard~K. Guy and Richard~J. Nowakowski.
\newblock Unsolved problems in combinatorial games.
\newblock In R.~J. Nowakowski, editor, {\em More Games of No Chance}, pages
  457--473. Cambridge University Press, 2002.

\bibitem[Sch78]{Schaefer-1977}
Thomas~J. Schaefer.
\newblock On the complexity of some two-person perfect-information games.
\newblock {\em Journal of Computer and System Sciences}, 16(2):185--225, 1978.

\bibitem[USM02]{US-mint}
United States~Mint\Bibkey USM.
\newblock The {U}nited {S}tates government to sell the famed 1933 double eagle,
  the most valuable gold coin in the world.
\newblock
  \url{https://www.usmint.gov/news/press-releases/20020207-the-united-states-government-to-sell-the-famed-1933-double-eagle-the-most-valuable-gold-coin-in-the-world},
  February 2002.

\end{thebibliography}
\bibliographystyle{alpha}

\end{document}